\documentclass[12pt,draftclsnofoot,onecolumn,peerreview]{IEEEtran}

\makeatletter
\long\def\@makecaption#1#2{\ifx\@captype\@IEEEtablestring%
\footnotesize\begin{center}{\normalfont\footnotesize #1}\\
{\normalfont\footnotesize\scshape #2}\end{center}%
\@IEEEtablecaptionsepspace \else \@IEEEfigurecaptionsepspace
\setbox\@tempboxa\hbox{\normalfont\footnotesize {#1.}~~ #2}%
\ifdim \wd\@tempboxa >\hsize%
\setbox\@tempboxa\hbox{\normalfont\footnotesize {#1.}~~ }%
\parbox[t]{\hsize}{\normalfont\footnotesize \noindent\unhbox\@tempboxa#2}%
\else \hbox
to\hsize{\normalfont\footnotesize\hfil\box\@tempboxa\hfil}\fi\fi}
\makeatother

\usepackage{graphicx}
\graphicspath{{figs/}}
\usepackage{subfigure}
\usepackage{amsthm}
\usepackage{amssymb}
\usepackage{amsmath}
\usepackage{algorithm}
\usepackage{algorithmic}
\usepackage{graphicx,color,psfrag}
\usepackage{cite}
\usepackage{enumerate}

\newcommand{\eg}{\textit{e.g.},}
\newcommand{\ie}{\textit{i.e.},}

\newcommand{\eqend}{\,.}
\newcommand{\iid}{IID}

\newcommand{\pr}[1]{\mathbb{P}\left\{ #1 \right\}}
\newcommand{\E}[1]{\mathbb{E}\left[#1\right]}

\newcommand{\com}[1]{\{#1\}}

\newtheorem{define}{Definition}
\newtheorem{lemma}{Lemma}
\newtheorem{theorem}{Theorem}

\newcommand{\srarq}{\mbox{SR-ARQ}}
\newcommand{\srfec}{\mbox{SR-FEC}}

\newcommand{\arq}{_\mathrm{\scriptscriptstyle ARQ}}
\newcommand{\fec}{_\mathrm{\scriptscriptstyle FEC}}

\begin{document}

\title{On the Scalability of Reliable Data Transfer in High Speed Networks}
\author{
    Majid Ghaderi$^{\dag}$
    \thanks{$^{\dag}$Department of Computer Science, University of Calgary, Email: {\tt mghaderi@ucalgary.ca}}
    and Don Towsley$^{\ddag}$
    \thanks{$^{\ddag}$Department of Computer Science, University of Massachusetts Amherst,
    Email: {\tt towsley@cs.umass.edu} }
}

\maketitle

\begin{abstract}
This paper considers reliable data transfer in a high-speed network (HSN) in which the per-connection capacity is very large. We focus on sliding window protocols employing {\em selective repeat} for reliable data transfer and study two reliability mechanisms based on ARQ and FEC. The question we ask is which mechanism is more suitable for an HSN in which the scalability of reliable data transfer in terms of receiver's buffer requirement and achievable delay and throughput is a concern. To efficiently utilize the large bandwidth available to a connection in an HSN, sliding window protocols require a large transmission window. In this regime of large transmission windows, we show that while both mechanisms achieve the same asymptotic throughput in the presence of packet losses, their delay and buffer requirements are different. Specifically, an FEC-based mechanism has delay and receiver's buffer requirement that are asymptotically smaller than that of an ARQ-based selective repeat mechanism by a factor of $\log W$, where $W$ is the window size of the selective repeat mechanism. This result is then used to investigate the implications of each reliability mechanism on protocol design in an HSN in terms of throughput, delay, buffer requirement, and control overhead.
\end{abstract}


\IEEEpeerreviewmaketitle

\section{Introduction}

The role of a reliable data transfer protocol is to deliver data from a traffic source to a receiver such that no data packet is lost or duplicated, and all packets are delivered to the receiving application in the order in which they were sent\footnote{In this paper, we only consider unicast protocols.}. Currently, TCP is the dominant protocol for reliable transmission of data in the Internet. The reliability mechanism in TCP is based on a \emph{sliding window} protocol in which the set of packets that can be transmitted by the sender at any time instant is determined by a logical window over the stream of incoming packets. The size of the window is governed by the TCP congestion control algorithm known as additive increase multiplicative decrease (AIMD). The poor performance of AIMD in a high-speed network (HSN) has been subject to numerous studies, and consequently, several modifications have been proposed to remedy its shortcomings~\cite{li07tcp}.

In this work, we argue that another fundamental problem of TCP in an HSN is its reliability mechanism. This problem is particularly manifested in an HSN where there is a non-zero probability of packet loss, for example, due to small router buffers~\cite{app04}. It has been shown that high-speed TCP variants achieve poor bandwidth utilization due to small buffers unless the loss probability at routers is maintained at non-negligible levels, \eg\ $1\%$ (see~\cite{gu07}). This results in a large number of lost packets that need to be retransmitted by TCP, causing out-of-order packets, and consequently long delays and a large buffer requirement at the receiver.

We consider two sliding window reliable data transfer protocols based on Automatic Repeat Request (ARQ) and Forward Error Correction (FEC). A sliding window protocol allows the sender to have multiple on-the-fly packets, which have been transmitted but not acknowledged yet. Typically, the channel between the sender and receiver is \emph{lossy}, and hence some packets may not reach the receiver. To implement reliability, the receiver sends feedback to the sender over a feedback channel. In this work, we assume that the sender and receiver implement the \emph{selective repeat} protocol (SR) for reliable transmission of packets. In the selective repeat protocol, the sender can transmit multiple packets before waiting for feedback from the receiver, while the receiver accepts every packet that arrives at the receiver~\cite{kurose}.

There are several important performance measures associated with the selective repeat protocol: throughput, buffer occupancy at the sender and receiver, packet delay, and protocol overhead. For reliable data transfer, at the receiver, packets have to be delivered to the application \emph{in-order}. Packets can arrive at the receiver out-of-order, for example, some packets may be lost causing gaps in the sequence of received packets. Thus, the receiver may need to buffer the arriving packets for some time until they can be delivered in-order. An important metric in analyzing the selective repeat protocol is the so-called \emph{re-sequencing delay} at the receiver's buffer. The re-sequencing delay is the delay between the time when a packet is received at the receiver and the time it is being delivered to the application. The set of packets in the receiver's buffer is called the re-sequencing buffer occupancy. The goal of this paper is to analyze the receiver's buffer occupancy, packet delay and protocol overhead of the selective repeat protocol based on ARQ and FEC reliability mechanisms in a setting where both protocols achieve their respective maximum throughput.

Traditionally, SR is coupled with ARQ to recover from packet losses. In this paper, in addition to the traditional SR, we consider another variation of the protocol in which ARQ is augmented with FEC at the sender to enable the receiver to recover from packet losses more efficiently. The intuition is that, with FEC, packets do not have individual identities, and hence as long as a sufficient number of them arrive at the receiver, the receiver is able to decode and deliver packets in-order to the application. In contrast, with pure ARQ, transmitted packets have unique identities, and hence the receiver has to receive at least one copy of every packet before it can deliver the packets in-order to the application. As shown in this paper, it takes a longer time and requires a larger buffer space for the receiver to collect a copy of every packet as opposed to just collecting a given number of packets.

Throughout the paper, we use the term ``high-speed network'' to refer to a network where the per-connection capacity approaches infinity regardless of the number of connections in the network. We show that in an HSN, reliability mechanisms based on FEC are more suitable as their buffer and delay performance is superior to that of ARQ based reliability mechanisms, while asymptotically achieving the same throughput. Specifically, we show that the ARQ and FEC based mechanisms, respectively, require $\Theta(W \log W)$ and $\Theta(W)$ buffer space, and achieve $\Theta(\log W)$ and $\Theta(1)$ delay, where $W$ denotes the window size in the selective repeat protocol.

We note that employing FEC based reliability mechanisms in TCP has been proposed in the literature. LT-TCP was proposed in~\cite{kk05} to augment TCP with adaptive FEC to cope with high packet loss rates in wireless networks. TCP/NC proposed in~\cite{medard09} incorporates network coding as a middleware below TCP to deal with packet losses. Our focus in this paper is to characterize the asymptotic performance of ARQ and FEC based mechanisms in HSNs to determine which mechanism is more scalable when per-connection capacity is very large.

Our contributions can be summarized as:
\begin{itemize}
\item We derive closed-form expressions for the performance of ARQ and FEC based
reliability mechanisms in terms of throughout, receiver's buffer and
delay for an arbitrary window size $W$.

\item We characterize the asymptotic performance of the ARQ and FEC based reliability mechanisms
as the per-connection capacity, and consequently the window size $W$, grows to infinity.

\item The asymptotic results are used to compare the performances of the two mechanisms in terms of throughput, delay, buffer requirement and protocol overhead.
\end{itemize}

The rest of this paper is organized as follows. In Section~\ref{s:rdt}, we describe the reliability mechanisms considered in the paper. Sections~\ref{s:exact} and~\ref{s:asymptotic} are dedicated, respectively, to the exact and asymptotic analysis of ARQ and FEC based reliability mechanisms. Several extensions and implications of our results on protocol design are discussed in Section~\ref{s:ext}. Section~\ref{s:related} reviews some related work, while our concluding remarks are presented in Section~\ref{s:conc}.

\section{Reliable Data Transfer Protocols}
\label{s:rdt}

In the following subsections, we describe the operation of two
selective repeat protocols based on ARQ and FEC for reliable data
transfer in lossy networks. Our objective is to study the asymptotic performance of these protocols, and hence we focus on the basic operation of each protocol and ignore complications that arise due to variable round-trip-times, channel reordering and imperfect feedback.

In traditional SR, the window size $W$ is used to limit the buffer space required at the sender and receiver by limiting the difference between the greatest and smallest sequence numbers of outstanding packets (packets that have been transmitted but not acknowledged yet) to $W$. Such a mechanism requires exactly $W$ buffer space at the sender and receiver but has a throughput penalty, and thus is unable to fully utilize the available network capacity. For instance, if every packet in a window is acknowledged except the first packet, the sender is not allowed to transmit any new packets since the window is exhausted. In contrast, we allow our protocols to have up to $W$ outstanding packets regardless of the range of the sequence numbers of the transmitted packets. Clearly, this approach allows SR to fully utilize the available capacity, which is desirable in a HSN, but requires larger buffer space at the sender and receiver. Our objective in this paper is to characterize the increased buffer requirement and corresponding packet delay at the receiver.

\subsection{ARQ Based Selective Repeat Protocol}
The ARQ-based protocol, referred to as \emph{SR-ARQ}, employs ARQ for
reliable data transmission. The sender continuously transmits new
packets in increasing order of sequence numbers as long as ACKs are
received for the transmitted packets. For every packet arriving at
the receiver, the receiver sends an ACK to the sender over the
feedback channel. The ACK packet contains the sequence number of the
packet for which the ACK is being transmitted. After transmitting a
packet, the sender transmits up to $W-1$ subsequent packets (new or
retransmitted). If an ACK arrives at the sender during this period,
the corresponding packet is removed from the sender buffer and a new
packet is subsequently transmitted. If time-out occurs, \ie\ no ACK
arrives during this time, then the same packet is subsequently
transmitted again.
The receiver buffers out-of-order packets for later in-order delivery to the application. A packet is delivered to the application and removed from the receiver's buffer only when all packets with smaller sequence numbers have been released from the buffer.

Algorithm~\ref{a:arq} summarizes the main events and their corresponding actions at the sender and receiver under \srarq.

\begin{algorithm}
\caption{\srarq\ events and actions.}
\textbf{Sender:} \com{$s^*:$ next packet sequence number}
\begin{itemize}
    \item \textbf{ACK($s$):} \com{arrival of ACK for packet $s$}
    \begin{enumerate}
        \item mark packet $s$ as ACKed
        \item remove in-order ACKed packets from buffer
        \item transmit packet $s^*$ and start a timer for it
        \item $s^* \leftarrow s^*+1$
    \end{enumerate}

    \item \textbf{Timeout($s$):} \com{time-out for packet $s$}
    \begin{enumerate}
        \item retransmit packet $s$ and start a timer for it
    \end{enumerate}
\end{itemize}
\textbf{Receiver:}
\begin{itemize}
    \item \textbf{Receive($s$):} \com{arrival of packet $s$}
    \begin{enumerate}
        \item transmit ACK $s$
        \item copy packet $s$ to buffer
        \item remove in-order packets from buffer and deliver them to application
    \end{enumerate}
\end{itemize}
\label{a:arq}
\end{algorithm}

\subsection{FEC Based Selective Repeat Protocol}
The FEC-based protocol, referred to as \emph{SR-FEC}, is a selective repeat protocol that implements a \emph{rateless}~\cite{mackay} coding algorithm for packet transmission. At the sender, packets are divided into consecutive \emph{coding blocks} for transmission to the receiver. Each coding block contains $B>1$ packets, which are coded together to form coded packets. A consecutive sequence number is assigned to each block to uniquely identify coding blocks. Every coded packet carries its corresponding block sequence number in its header in addition to a packet sequence number that is unique within its corresponding block. Thus, the pair of block and packet sequence numbers uniquely identifies a packet.

While \srfec\ considered in this paper is independent of the specific coding technique used to generate coded packets, for the sake of concreteness, we assume that the sender performs {\em random fountain coding}~\cite{mackay} to generate coded packets. In random fountain coding, at every transmission instance, a subset of the packets in the current coding block is chosen randomly and the packets in the subset are XOR'ed together to create a coded packet. The coded packet is then transmitted to the receiver. Random fountain codes are a special class of fountain codes where the degree distribution of coded packets is uniform. While their decoding complexity is higher than LT and Raptor codes, they exhibit similar error correction performance, which is the reason for their consideration in this paper.

The receiver decodes packets in a block-by-block manner. To decode a
block, $B$ \emph{independent} coded packets from that block are required. Upon
receiving a coded packet, the receiver infers which original packets
were XOR'ed to create the newly received coded packet. When the
block+packet sequence numbers are unique, the coding and decoding
algorithms at the sender and receiver can be synchronized, hence
avoiding the need to include any extra information in packet headers
to help infer the identity of the XOR'ed packets. This can be
implemented, for example, by synchronizing a random number generator
at the receiver with the random number generator used at the sender
for choosing the set of XOR'ed packets.

Similar to \srarq, the sender continuously transmits coded packets in
increasing order of coding block sequence number as long as ACKs are
received. For every \emph{innovative} coded packet (\ie\ a packet
that is independent of the previously received coded packets for the
corresponding block) arriving at the receiver, the receiver sends an
ACK to the sender over the feedback channel. The ACK only contains
the sequence number of the coding block for which the ACK is being
transmitted. After transmitting a packet, the sender transmits up to
$W-1$ subsequent coded packets. The arrival of an ACK arrives at the sender
during this period triggers a new coded packet.
The new packet is constructed from the coding block that contains the
packet. If a time-out occurs, \ie\ no ACK arrives, then a coded packet
from the coding block that caused the time-out is subsequently
transmitted. When $B$ ACKs are received for a coding block then it is
released from the sender window.

Algorithm~\ref{a:fec} summarizes the main events and their corresponding actions at the sender and receiver under \srfec.

\begin{algorithm}
\caption{\srfec\ events and actions.}
\textbf{Sender:} \com{$b^*:$ next block sequence number}
\begin{itemize}
    \item \textbf{ACK($b$):} \com{ACK for block $b$}
    \begin{enumerate}
        \item if block $b$ has $B$ ACKs then mark it as ACKed
        \item remove in-order ACKed blocks from buffer
        \item transmit a coded packet for block $b^*$ and\\ start a timer for it
        \item if $b^*$ has $B$ pending and received ACKs then\\ $b^* \leftarrow b^*+1$
    \end{enumerate}
    \item \textbf{Timeout($b$):} \com{time-out for block $b$}
    \begin{enumerate}
        \item retransmit a coded packet for block $b$ and start a timer for it
    \end{enumerate}
\end{itemize}
\textbf{Receiver:}
\begin{itemize}
    \item \textbf{Receive($b$):} \com{arrival of coded packet for block $b$}
    \begin{enumerate}
        \item transmit ACK $b$
        \item copy the packet to receiver buffer for block $b$
        \item if block $b$ is full rank then decode block $b$
        \item remove in-order blocks from buffer and deliver them to application
    \end{enumerate}
\end{itemize}
\label{a:fec}
\end{algorithm}

\section{Exact Performance Analysis} \label{s:exact}
An important metric in characterizing the performance of ARQ
protocols is the \emph{expected number of retransmissions} required
at the sender until a packet arrives successfully at the receiver.
Our focus in this section is on deriving exact expressions for this
metric. Using this metric, we show how it can be used to compute the
receiver's buffer occupancy and packet delay for \srarq\ and \srfec.

\subsection{Assumptions}
We consider a heavy traffic situation, in which packets arrive at the
sender from an infinite source. For any arrival process, this model
provides upper bounds on the receiver buffer requirements and packet
delay. It is assumed that all packets are of the same size, and that
the packet loss process is independent of the packet size. This
represents a scenario where packet losses are primarily due to
congestion at routers along the path between the sender and receiver.
Specifically, we assume that packet losses form a Bernoulli process
with mean $p$.

The analysis presented in this paper assumes that the sender window
size is fixed. We note that in TCP, the sender window size is dynamic
and changes in response to network congestion situation and available
buffer at the receiver. In such cases, our model is not accurate and
should be extended to dynamic window sizes. However, in situations
where TCP congestion window is stable for long periods of time, our
results can be applied. For example, networks with low loss rate
(which are of interest in this work), if the TCP congestion window is
limited by the receiver's advertised window, then the sender window
size remains constant for long periods of time. In a sense, our model
assumes decoupled reliability from congestion control, as advocated
for in~\cite{sundaresan03atp,kohler06dccp,gu07}.

To simplify the analysis, as commonly assumed in the
literature~\cite{towsley79,konheim80arq,anagnos86,ros89}, we assume
that a reliable feedback channel exists between the sender and
receiver. The receiver sends feedback in the form of acknowledgement
packets (ACKs) to the sender, which are assumed to be delivered over
the feedback channel reliably. We also assume that the
round-trip-time between the sender and receiver is fixed and equal to
$R$.

\subsection{Performance Metrics}
In high-speed networks, to efficiently utilize network bandwidth, the
window size $W$ should grow proportional to the per-connection
capacity. This is simply stated as $W = R \cdot C$, where $C$ is the
per-connection capacity. Thus, as per-connection capacity increases,
the window size should increase as well resulting in increased buffer
occupancy at the sender and receiver. We assume that the sender and
receiver have sufficient buffer space to allow each protocol, namely
\srarq\ and \srfec, to operate efficiently.

With \srarq, every packet received at the receiver is a useful
packet, and hence the throughput of \srarq\ is given by $\rho = (1-p)
\, C$. However, with \srfec, only independent coded packets are
useful for the receiver to decode and recover the original packets.
To avoid transmitting coded packets that are not independent from the
previously transmitted packets, and hence wasting bandwidth, for
every active coding block, the sender memorizes the coded packets
transmitted so far for that coding block, for example, by keeping
track of the indices of the packet subsets used to construct the
coded packets. For a given coding block, the sender only sends
packets that are independent of the coded packets transmitted
previously. In this case, the throughput of \srfec\ is also given by
$\rho = (1-p) \, C$. The analysis presented in Sections~\ref{s:exact}
and~\ref{s:asymptotic} considers this case. Later, in
Section~\ref{s:ext}, we address the case when coded packets may be
dependent.

Let $Q$ and $D$ denote the re-sequencing buffer occupancy and delay
at the receiver, respectively. Using the Little's law, it follows
that
\begin{equation}
\begin{split}
    \E{D}
        = \frac{\E{Q}}{(1-p) C} = \frac{R}{(1-p) W} \, \E{Q}
    \eqend
\end{split}
\end{equation}
In the rest of the paper, we focus on computing $\E{Q}$ for both protocols.

\subsection{\srarq\ Analysis}
Consider a window of $W$ packets. Let $N_i$ denote the number of
times that the packet at position $i$, for $1 \le i \le W$, is
transmitted until it is received successfully at the receiver. It
follows that $N_i$ has a geometric distribution. In other words, the
probability that the packet at position $i$ is received successfully
in $n$-th transmission is given by
\begin{equation}
    \pr{N_i = n} = (1-p)p^{n-1}, \qquad n \ge 1
    \eqend
\end{equation}
Among the $W$ packets in the window, let $\ell$ denote the last
packet that is successfully received at the receiver, \ie\ the packet
that takes the most number of retransmissions. Let $N\arq$ denote the
number of transmissions required until $\ell$ is received at the
receiver. It follows that
\begin{equation}
    N\arq = \max_{1 \le i \le W} N_i
    \eqend
\end{equation}
Since packet losses are assumed to be independent, the following
relation holds for the probability distribution of the random
variable $N\arq$,
\begin{equation}
    \pr{N\arq \le n} = (1 - p^n)^W, \qquad n \ge 1
        \eqend
\end{equation}
Using the above expression, it follows that,
\begin{equation}
\begin{split}
    \E{N\arq}
        = \sum_{n=1}^{\infty} \pr{N \ge n}
        = \sum_{n=1}^{\infty} 1 - (1 - p^{n-1})^W
        \eqend
\end{split}
\end{equation}
Define $f(K)$ for $k \ge 1$ as,
\begin{equation}
    f(K) = \sum_{n=1}^{K} 1 - (1 - p^{n-1})^W,
\end{equation}
which, then yields,
\begin{equation}
    \E{N\arq} = \lim_{K \rightarrow \infty} f(K)
    \eqend
\end{equation}
Next, we focus on computing a closed-form expression for $f(K)$ that
can be used to compute $\E{N\arq}$:
\begin{equation}
\begin{split}
    f(K)
        &= K - \sum_{n=1}^{K} (1 - p^{n-1})^W \\
        &= K - \sum_{n=1}^{K} \sum_{i=0}^{W} \binom{W}{i} (-p^{n-1})^i\\
%        &= K - \sum_{n=1}^{K} \Big( 1 + \sum_{i=1}^{W} \binom{W}{i} (-1)^i (p^{n-1})^i \Big)\\
%        &= - \sum_{i=1}^{W} \binom{W}{i} (-1)^i \sum_{n=1}^{K} (p^i)^{n-1}\\
        &= - \sum_{i=1}^{W} \binom{W}{i} (-1)^i \frac{1 - (p^i)^K}{1 - p^i}
        \eqend
\end{split}
\end{equation}

Taking the limit of $f(K)$ as $K \rightarrow \infty$ yields the
following result for $\E{N\arq}$:
\begin{equation}
\label{e:EN-arq}
\begin{split}
    \E{N\arq}
        &= \lim_{K \rightarrow \infty} f(K)\\
        &= - \sum_{i=1}^{W} \binom{W}{i} \frac{(-1)^i}{1 - p^i}
        \eqend
\end{split}
\end{equation}

The above expressions are valid when $ 0\le p < 1$. If $p=1$, then no
packet is received at the receiver and hence the receiver's buffer
will be empty. Over the range $ 0\le p < 1$, as $p$ increases so does
the number of lost packets. As the number of lost packets increases,
the number of out-of-order packets in the receiver's buffer increases
as well. %Fig.~\ref{fig:N-vs-W-arq} shows the behavior of $\E{N}$ with
%respect to window size $W$ for various packet loss probabilities.
%\begin{figure}
%   \centering
%    \resizebox{0.4\textwidth}{!}{
%         \input{figs/N-vs-W-arq.tex} }
%    \caption{Expected no. of transmissions for \srarq\ as given by~\eqref{e:EN-arq}.}
%    \label{fig:N-vs-W-arq}
%\end{figure}

\subsection{\srfec\ Analysis}
With coding, packets are delivered to the application in a
block-by-block manner. For every coding block, the receiver needs to
receive $B$ independent coded packets in order to decode and recover
the original coding block. Let $N_i$ denote the number of transmitted
coded packets until the receiver decodes coding block $i$. The
probability that the receiver receives exactly $B$ coded packets
after the $n$-th packet transmission by the sender for coding block
$i$ is given by a negative binomial distribution, \ie
\begin{equation}
    \pr{N_i = n} = \tbinom{n-1}{B-1} (1-p)^B p^{n-B}, \quad n \ge B
\end{equation}
which is the probability that $B-1$ packets out of the first $n-1$
transmissions are received successfully and the last transmission is
a success too. Recall that, for the moment, we are assuming all coded
packets are independent following the mechanism described earlier.
Therefore, we obtain that,
\begin{equation}
\begin{split}
    \pr{N_i \le n} &= \sum_{i=B}^{n} \tbinom{i-1}{B-1} (1-p)^B p^{i-B},\\
                    &= (1-p)^B \sum_{i=0}^{n-B} \tbinom{B+i-1}{B-1} p^i
    \eqend
\end{split}
\end{equation}

Consider a typical window of $W$ packets consisting of $M = W/B$
coding blocks. Without loss of generality, we assume that $W$ is an
integer multiple of $B$. Among the $M$ coding blocks, let $\ell$
denote the last block that is decoded by the receiver, \ie\ the block
that takes the most number of transmissions until its corresponding
$B$ coded packets are received. Let $N\fec$ denote the number of
transmissions required until block $\ell$ is decoded by the receiver.
It is obtained that
\begin{equation}
    N\fec = \max_{1 \le i \le W} N_i
    \eqend
\end{equation}
Since packet losses are assumed to be independent, the following
relation holds for the probability distribution of the random
variable $N\fec$,
\begin{equation}
    \pr{N\fec \le n} = (1-p)^W \Big( \sum_{i=0}^{n-B} \tbinom{B+i-1}{B-1} p^i \Big)^M
    \eqend
\end{equation}
However, it seems unlikely that this would result in any simple
expression for $\E{N\fec}$. Note that $N\fec$ is the number of
transmissions per coding bock of size $B$. Thus, the expected number
of transmissions per packet is given by $\frac{1}{B} \E{N\fec}$.

\section{Asymptotic Performance Analysis} \label{s:asymptotic}
The exact expressions derived in the previous section do not provide
insight about the scaling behavior of the number of transmissions
with respect to $W$. In this section, we derive asymptotic
expressions for the expected number of transmissions for both \srarq\
and \srfec. Then, using our asymptotic results, we compute the
expected receiver's buffer occupancy.
In particular, we are interested in the asymptotic performance of the
protocols as the window size becomes very large, \ie\ $W \rightarrow
\infty$. For \srfec, we consider two scaling regimes: (1) the coding
block length remains constant, \ie\ $B \in \Theta(1)$, and (2) the
coding block length scales at the same rate as the window size, \ie\
$B \in \Theta(W)$.

The asymptotic analysis presented in this section utilizes techniques
from the extreme value theory~\cite{galambos}, which is briefly
described next.

\subsection{Extreme Value Theory}
Consider a sequence of \iid\ random variables $X_1, X_2, \ldots, X_n$
with a common cumulative distribution function $F$, \ie\
    $F(x) = \pr{X_i \le x}$.
Let the $X_i$'s be arranged in increasing order of magnitude and
denoted by $X_{(1)}, X_{(2)}, \dots, X_{(n)}$. These ordered values
of $X_i$'s are called \emph{order statistics}, $X_{(k)}$ being the
$k$th order statistics. Note that $X_{(1)}$ and $X_{(n)}$ are the
minimum and maximum and are called the \emph{extreme} order
statistics.

Of particular interest in this paper is the asymptotic behavior of
$X_{(n)}$ for large $n$. Since $X_i$'s are independent, the
cumulative distribution function of $X_{(n)}$ is given by
\begin{equation*}
\begin{split}
    \pr{X_{(n)} \le x}
        = \pr{X_{(1)} \le x, \ldots, X_{(n)} \le x}
        =  F(x)^n,
\end{split}
\end{equation*}
and, thus, for $F(x) < 1$, as $n \rightarrow \infty$,
\begin{equation*}
    \pr{X_{(n)} \le x} \rightarrow 0
    \eqend
\end{equation*}

While it is possible to study the exact distribution of $X_{(n)}$, in
many applications, only the behavior of the tail of the distribution
is of interest, which is generally difficult to characterize using
this exact form. Extreme value theory provides a framework to study
the limiting distribution of extreme order statistics.

\begin{theorem}[Extreme Value Theorem~\cite{galambos}]
If there exists a sequence of real constants $a_n > 0$ and $b_n$,
such that, as $n \rightarrow \infty$,
\begin{equation}
\label{e:evt}     
    \pr{\frac{X_{(n)} - b_n}{a_n} \le x} = F(a_n x + b_n)^n \rightarrow G(x),
\end{equation}
for some non-degenerate distribution function $G$, then the limit
distribution $G$ must be one of the three extreme value distributions
as follows.
\begin{align*}
        \text{Gumbel: } G(x) &= \exp(-e^{-x}), \quad -\infty < x < \infty;\\
        \text{Fr\'{e}chet: } G(x) &=
                            \begin{cases}
                                0, & x \le 0,\\
                                \exp(-x^{-\alpha}), & x > 0, \alpha > 0;
                            \end{cases}\\
        \text{Weibull: } G(x) &=
                            \begin{cases}
                                \exp(-(-x)^{\alpha}), & x \le 0, \alpha > 0,\\
                                1, & x > 0 \eqend
                            \end{cases}
\end{align*}
\end{theorem}

\begin{define}[Maximum Domain of Attraction~\cite{haan06evt}]
The maximum domain of attraction of a distribution $G$ is the set of all distributions $F$ that has $G$ as the limiting distribution. That is, as $n \rightarrow \infty$, constants $a_n > 0$ and $b_n$ exist such that~\eqref{e:evt} holds.
\end{define}

Define symbol $x^*$ as the upper bound of $X_1$, that is,
\begin{equation*}
    x^* = \sup\{x | F(x) < 1\}
    \eqend
\end{equation*}
\begin{define}[von Mises Function~\cite{haan06evt}]
Suppose $F(x)$ is a distribution function with density $f(x)$ which is positive and differentiable on a left neighborhood of $x^*$. If
\begin{equation}
\label{e:mises}
    \lim_{x \rightarrow x^*} \frac{d}{dx} \Big( \frac{1-F(x)}{f(x)} \Big) = 0,
\end{equation}
then $F$ is called a von Mises function.
\end{define}
If distribution function $F$ is a von Mises function then it belongs to the domain of attraction of the Gumbel distribution with a possible choice of the sequence of the normalizing constants $a_n$ and $b_n$ given by
\begin{equation}
\label{e:ab}
    b_n = F^{-1}(1 - \frac{1}{n}), \quad
    a_n = \frac{1 - F(b_n)}{f(b_n)}
    \eqend
\end{equation}
Noting that for the
standard Gumbel distribution, the mean is given by $\gamma$, where
$\gamma\approx0.57721$ is the Euler's constant, it is obtained that, as $n \rightarrow \infty$,
\begin{equation}
\label{e:mean}
    \E{X_{(n)}} = a_n \gamma + b_n
    \eqend
\end{equation}

\subsection{\srarq\ Analysis}
Recall that $N_i$, for $1 \le i \le W$, has a geometric distribution.
Thus, $N\arq$ is the maximum of $W$ independent and geometrically
distributed random variables with the mean $1/(1-p)$. We will show
that $\E{N\arq} \in \Theta(\log W)$.

\begin{theorem}
\label{t:arq}
    $\E{N\arq} \in \Theta(\log W)$.
\end{theorem}
\begin{proof}
We approximate each geometric random variable $N_i$ by an
exponentially distributed random variable $X_i$ with an appropriate
rate $\lambda$. To obtain the same numerical values for the probability
distribution of $N_i$ and $X_i$, the following relation must be
satisfied:
\begin{equation}
    \pr{N_i \le n} = \pr{X_i \le n},
\end{equation}
which yields $\lambda=-\ln p$. Define $X$ as follows,
\begin{equation}
    X = \max_{1 \le i \le W} X_i
    \eqend
\end{equation}
Next, we show that the approximation error introduced as the result
of approximating $\E{N\arq}$ by $\E{X}$ is bounded by a constant,
\ie\ $\E{N\arq} = \E{X} + \Theta(1)$. To compute the approximation
error, denoted by $\epsilon$, we note that
\begin{equation*}
    \E{X_i} = \sum_{n=0}^{\infty} \int_{n}^{n+1} \pr{X_i > x}\,dx
    \eqend
\end{equation*}
Therefore,
\begin{equation}
\begin{split}
    \epsilon = \E{N_i} - \E{X_i}
             &= \sum_{n=0}^{\infty} p^n + \frac{1}{\lambda}
                \sum_{n=0}^{\infty} (e^{-(n+1)\lambda} - e^{-n\lambda})\\
             = \frac{1}{1-p} + \frac{1}{\ln p}, \quad 0 \le p \le 1
    \eqend
\end{split}
\end{equation}
Thus, the problem is reduced to computing the maximum of $W$ \iid\
exponentially distributed random variables with rate $\lambda=-\ln p$.

\begin{lemma}
The distribution function of an exponentially distributed random variable is a von Mises function.
\end{lemma}
\begin{proof}
See~\cite{haan06evt}.
\end{proof}
Therefore, there exist constants $a_W > 0$ and $b_W$ so that
$(X-b_W)/a_W$ converges in distribution to the standard Gumbel
distribution, as $W \rightarrow \infty$. A possible choice for $a_W$
and $b_W$ is given by~\eqref{e:ab}, which results in (\eg\
see~\cite{haan06evt}),
\begin{equation*}
    a_W = \lambda^{-1},  \quad b_W = \lambda^{-1} (\ln W)
    \eqend
\end{equation*}
Substituting the above values in~\eqref{e:mean} yields the following
expression for $\E{X}$,
\begin{equation}
    \E{X} = \frac{\ln W + \gamma}{-\ln p},
\end{equation}
as $W \rightarrow \infty$. Given that $\E{N\arq}$ is within a
constant of $\E{X}$, the theorem follows.
\end{proof}

Let $\E{Q\arq}$ denote the expected buffer occupancy at the receiver
under \srarq. In the following, we establish upper and lower bounds
on $\E{Q\arq}$, and show that $\E{Q\arq}$ scales as $W \log W$.

\begin{lemma}
\label{l:up}
    $\E{Q\arq} \in O(W \log W)$.
\end{lemma}
\begin{proof}
Let $\ell$ denote the packet that requires the most number of
retransmissions in a window. In the worst-case, $\ell$ is the first
packet in the window, \ie\ the packet with the smallest sequence
number, and all packets with sequence numbers greater than that of
$\ell$ are always received correctly at the receiver. In this case,
between every retransmission of $\ell$, $W-1$ new packets with
greater sequence numbers are transmitted. The receiver has to buffer
all these new packets until the last retransmission of $\ell$
arrives. As the result, we obtain that
\begin{equation}
    Q\arq \le (W-1) \cdot N\arq,
\end{equation}
and, therefore,
\begin{equation}
    \E{Q\arq} \le (W-1) \cdot \E{N\arq}
    \eqend
\end{equation}
It then follows that $\E{Q\arq} \in O(W \log W)$.
\end{proof}

\begin{lemma}
\label{l:low}
    $\E{Q\arq} \in \Omega(W \log W)$.
\end{lemma}
\begin{proof}
To compute a lower bound, consider the case when $\ell$ is the last
packet in the window, \ie\ has the largest sequence number in the
window. As time progresses, some of the packets with smaller sequence
number are received at the receiver in-order and consequently
delivered to the application. Every new packet that arrives at the
receiver with sequence number greater than that of $\ell$ will be
placed in the buffer. In the best-case, packets arrive at the
receiver in-order so that every time a packet is received it can be
delivered to the application and its corresponding buffer is
released.

Next, consider the time (in terms of the number of packet
transmissions) it takes for the first $W/2$ packets in the window to
be received and delivered to the application, assuming that all these
packets arrive before packets in the second half of the window (as we
are considering the best-case). We show that this time, denoted by
$\tau$, is upper bounded by a constant independent of the window size
$W$. Thus, after some constant initial time $\tau$, for every
retransmission of $\ell$, at least $W/2$ new packets with higher
sequence numbers arrive at the receiver that need to be buffered.

We are interested in the expected number of transmissions required to
receive the first $W/2$ packets in the best-case. For $n$ \iid\
exponentially distributed random variables with the rate $\lambda$,
it is well-know that the time until the first occurrence is given by
$1/(n\lambda)$. Thus, the expected time until $W/2$ packets arrive at
the receiver is given by,
\begin{equation}
    \E{\tau} = \frac{1}{\lambda} \big( \frac{1}{W} + \cdots + \frac{1}{W/2} \big) \le \frac{1}{\lambda}
    \eqend
\end{equation}
Putting it together, we have,
\begin{equation}
    Q\arq \ge \frac{W}{2} \cdot (N\arq - \tau),
\end{equation}
and, therefore,
\begin{equation}
    \E{Q\arq} \ge \frac{W}{2} \cdot (\E{N\arq} - \frac{1}{\lambda})
    \eqend
\end{equation}
It then follows that $\E{Q\arq} \in \Omega(W \log W)$.
\end{proof}

\begin{theorem}
\label{t:arq}
    $\E{Q\arq} \in \Theta(W \log W)$.
\end{theorem}
\begin{proof}
It follows from Lemmas~\ref{l:up} and~\ref{l:low}.
\end{proof}

\subsection{\srfec\ Analysis}
It is well-known that the performance of rateless coding schemes
depends on the size of the coding block. To capture this dependency,
we consider two scaling regimes as follows:
\begin{enumerate}[I.]
\item The size of the coding block remains constant as $W$ grows, \ie\ $B \in \Theta(1)$.
In this case, the number of coding blocks $M$ grows at the same rate
as $W$.
\item The size of the coding blocks grows at the same rate as $W$. In
this case, the number of coding blocks remains constant, \ie\ $M \in
\Theta(1)$, while the size of the coding blocks grows to infinity.
\end{enumerate}

In the rest of this section, we analyze the performance of \srfec\ in
these two scaling regimes. We show that, the asymptotic receiver's
buffer occupancy grows by a factor of $\log W$ slower in the second
scaling regime compared to the first scaling regime.

\bigskip
\noindent\textbf{Scaling Regime I: Constant Coding Block Size\\}
We show that, in this case, the receiver's buffer occupancy has the same scaling as with \srarq.
\begin{theorem}
\label{t:fecB}
    $\E{N\fec} \in \Theta(\log W)$.
\end{theorem}
\begin{proof}
Recall that the number of transmissions $N_i$ for coding block $i$
has a negative binomial distribution. In an alternative but
equivalent representation, $N_i$ can be considered as the sum of $B$
IID geometric random variables with the mean $1/(1-p)$. Each
geometric random variable represents the number of transmissions at
the sender until the next required coded packet arrives at the
receiver. That is,
\begin{equation}
\label{e:Ni-fec-expand}
    N_i = N_{i1} + \cdots + N_{iB},
\end{equation}
where, $N_{ik}$ is the number of transmissions of packets for coding
block $i$ at the sender until the receiver receives the $k$th coded
packet, given that it has already received $k-1$ coded packets. For
geometric random variables we have
\begin{equation}
    \pr{N_{ik} = n} = (1-p)p^{i-n}
    \eqend
\end{equation}
Following the approximation method described in the previous section,
we approximate each geometric random variable $N_{ik}$ with an
exponentially distributed random variable $X_{ik}$ with the rate
$\lambda = -\ln p$. Define $X_i$ and $X$ as follows:
\begin{equation}
    X_i = X_{i1} + \cdots + X_{iB},
\end{equation}
and, for $M$ coding blocks in a window,
\begin{equation}
    X = \max_{1 \le i \le M} X_i
    \eqend
\end{equation}
Our goal is to approximate $\E{N\fec}$ with $\E{X}$. From the
previous section, the approximation error is given by
\begin{equation}
\label{e:err}
    \epsilon = B \cdot \big( \frac{1}{1-p} + \frac{1}{\ln p} \big)
    \in  \Theta(B),
\end{equation}
which is constant.

Next, we turn our attention to computing $\E{X}$. Notice that $X_i$
is the summation of $B$ IID exponentially distributed random
variables. Thus, $X_i$ is Erlang distributed with the following
cumulative distribution function:
\begin{equation*}
    F(x) = \pr{X_i \le x} = 1 - e^{-\lambda x} \sum_{k=0}^{B-1} \frac{(\lambda x)^k}{k!}
    \eqend
\end{equation*}

\begin{lemma}
The distribution function of the sum of $B$ IID exponentially distributed random variables is a von Mises function.
\end{lemma}
\begin{proof}
We need to show that condition~\eqref{e:mises} holds for the Erlang distribution. Equivalently, the following condition should be satisfied,
\begin{equation}
\label{e:alt}
    \lim_{x \rightarrow \infty} \frac{f'(x) (1-F(x))}{f^2(x)} = -1,
\end{equation}
where $f$ and $f'$ denote the probability density function of Erlang distribution and its derivative respectively. To this end, we have
\begin{equation*}
    f(x) = \frac{\lambda^B x^{B-1} e^{- \lambda x}}{(B-1)!}, \quad
    f'(x) = \big(\frac{B-1}{x} - \lambda \big) f(x)
    \eqend
\end{equation*}
After substitution in~\eqref{e:alt}, the lemma follows.
\end{proof}

Therefore, there exist constants $a_M > 0$ and $b_M$ so that
$(X-b_M)/a_M$ converges in distribution to the standard Gumbel
distribution, as $M \rightarrow \infty$. A possible choice for $a_M$
and $b_M$ is given by~\eqref{e:ab}. Specifically, $b_M$ can be
computed by solving the equation $b_M = F^{-1}(1-1/M)$. Using the
following representation of the Erlang distribution,
\begin{equation*}
    F(x) = 1 - e^{- \lambda x} \Big[ \frac{(\lambda x)^{B-1}}{(B-1)!} + \Theta(x^{B-2}) \Big],
\end{equation*}
it is obtained that,
\begin{equation*}
    \frac{1}{M} = e^{- \lambda b_M} \Big[ \frac{(\lambda b_M)^{B-1}}{(B-1)!} + \Theta(b_M^{B-2}) \Big]
    \eqend
\end{equation*}
By taking the logarithm of both sides of the equation, and using the relation
    $\ln(r+t) = \ln r + \ln(1 + \frac{t}{r})$, we have,
\begin{equation*}
    \lambda b_M = \ln M + (B-1) \ln (\lambda b_M) - \ln (B-1)! + o(1),
\end{equation*}
which yields the following expression for $b_M$,
\begin{equation*}
    b_M = \lambda^{-1} (\ln M + (B-1) \ln\ln M - \ln (B-1)!)
    \eqend
\end{equation*}
By substituting $b_M$ in~\eqref{e:ab}, it is then obtained that $a_M
= \lambda^{-1}$. Substituting $a_M$ and $b_M$ in~\eqref{e:mean}, the
following expression is obtained for $\E{X}$,
\begin{equation}
    \E{X} = \frac{\ln W + (B-1) \ln\ln W + \Theta(B \ln B)}{-\ln p},
\end{equation}
as $W \rightarrow \infty$. Given the constant approximation
error~\eqref{e:err}, it is obtained that $\E{N\fec} \in \Theta(\log
W)$. Since $B$ is constant, the number of per packet transmissions,
\ie\ $\tfrac{1}{B} \E{N\fec}$, also is in $\Theta(\log W)$.
\end{proof}

Using $\E{N\fec}$, we show that $\E{Q\fec} \in \Theta(W \log W)$. The
proof is similar to the proof presented for $\E{Q\arq}$ in the
previous section, and hence only a sketch of the proof is presented
in the following lemmas.

\begin{lemma}
\label{l:up-coding}
    $\E{Q\fec} \in O(W \log W)$.
\end{lemma}
\begin{proof}
In the worst-case, the coding block that requires the most number of
packet transmissions is the block with the smallest sequence number
in the window. For every packet transmission for this block, $W-1$
other packets are transmitted for the proceeding blocks. Thus, the
expected receiver's buffer occupancy is upper bounded by $W \log
W$.
\end{proof}

\begin{lemma}
\label{l:low-coding}
    $\E{Q\fec} \in \Omega(W \log W)$.
\end{lemma}
\begin{proof}
Assume that the coding block with the largest sequence number is at
the end of the window. Thus, all coding blocks before it can be
potentially delivered to the application if they are receiver at the
receiver in-order. Consider the time it takes for the first $M/2$
coding blocks to be delivered to the application. In the best-case,
these coding blocks arrive at the receiver in an increasing order.
Thus, the best-case time to receive the first $M/2$ blocks is upper
bounded by the time it takes to receive $W/2$ coded packets belonging
to the first half of the window. In the previous section, we showed
that this time is upper bounded by a constant. Thus, the expected
receiver's buffer occupancy is lower bounded by $W \log W$.
\end{proof}

\begin{theorem}
\label{t:fec1}
    $\E{Q\fec} \in \Theta(W \log W)$.
\end{theorem}
\begin{proof}
The theorem follows from Lemmas~\ref{l:up-coding}
and~\ref{l:low-coding}.
\end{proof}

\bigskip
\noindent\textbf{Scaling Regime II: Growing Coding Block Size\\}
Recall that in this case, we have $B \in \Theta(W)$ so that $M$ is a
constant. We show that, in this case, the expected number of
transmissions per packet is constant. Consequently, we show that
$\E{Q\fec} \in \Theta(W)$. The proof relies on the Lebesgue's Dominated Convergence Theorem stated below.

\begin{theorem}[Dominated Convergence Theorem]
Suppose that $X_n \rightarrow X$ almost surely as $n \rightarrow
\infty$. If there exists a random variable $Y$ having finite
expectation and $|X_n| \le Y$, for all $n$, then $\E{X_n} \rightarrow
\E{X}$, as $n \rightarrow \infty$
\end{theorem}

\begin{theorem}
    $\E{N\fec} \in \Theta(W)$.
\end{theorem}
\begin{proof}
Define $N$ as the number of transmissions per packet until the entire
window is received at the receiver, that is, $N = \tfrac{1}{B}
N\fec$. Using the expansion introduced in~\eqref{e:Ni-fec-expand}, we
are interested in characterizing $\E{N}$,
\begin{equation}
\label{eq:N}
    \E{N} = \frac{1}{B} \E{\max_{1\le i \le M} N_i} = \E{\max_{1\le i \le M} \frac{N_i}{B}}
    \eqend
\end{equation}
Now, we rewrite \eqref{eq:N} as follows,
\begin{equation}
\begin{split}
    \E{N\fec}
            &= \E{\max_{1\le i \le M} \frac{N_i}{B}}\\
            &= \E{\max_{1\le i \le M} \frac{N_{i1} + \cdots + N_{iB}}{B}}
            \eqend
\end{split}
\end{equation}

Note that the $N_{ik}$'s are IID and hence, by applying the strong law of large numbers, $\tfrac{N_i}{B} \rightarrow \frac{1}{1-p}$, as $B \rightarrow \infty$, almost surely.
Thus, $N \rightarrow \frac{1}{1-p}$ almost surely as $B \rightarrow
\infty$. Define random variable $Y$ as $Y = \sum_{1\le i \le M} \frac{N_i}{B}$. Notice that $N \le Y$ for all values of $B$, and that $\E{Y} = \tfrac{M}{1-p} < \infty$. Therefore, the Lebesgue's Dominated Convergence Theorem can be applied, which yields,
\begin{equation}
\label{eq:nc}
\begin{split}
    \lim_{B \rightarrow \infty} \E{N}
        &= \E{\max_{1\le i \le M} \lim_{B \rightarrow \infty} \frac{N_{i1} + \cdots + N_{iB}}{B}}, \\
        &= \E{\max_{1\le i \le M} \frac{1}{1-p}}
        = \frac{1}{1-p}
        = \Theta(1)
        \eqend
\end{split}
\end{equation}
It then follows that $\E{N\fec} = B/(1-p)$, which completes the
proof.
\end{proof}

\begin{theorem}
\label{t:fec2}
    $\E{Q\fec} \in \Theta(W)$.
\end{theorem}
\begin{proof}
Notice that in this case, all coding blocks are decoded at the same
time. Thus, the expected buffer occupancy is given by $W/(1-p)$,
which established the theorem.
\end{proof}

\section{Extensions and Implications} \label{s:ext}

\begin{table*}
\caption{Asymptotic Performance of Selective Repeat Protocols.}
\label{t:summary}
\centering
\begin{tabular}{|l|c|c|c|c|}
\hline
Protocol    &   Throughput  &   Reseq. Buffer &   Reseq. Delay  & Feedback Overhead\\
\hline
\hline
\srarq\     &   $(1-p) C$   &   $\Theta(W \log W)$  &   $\Theta(\log W)$  &  $\Theta(\log W)$\\
\hline
\srfec, $B=\Theta(1)$ & $(1-p) C$ & $\Theta(W \log W)$  &   $\Theta(\log W)$   &  $\Theta(\log W)$ \\
\hline
\srfec, $B=\Theta(W)$ & $(1-p) C$ & $\Theta(W)$  &   $\Theta(1)$   &  $\Theta(1)$ \\
\hline
\end{tabular}
\end{table*}

\subsection{Lossy Feedback}
If an ACK packet is lost, the sender will time-out and retransmit a packet.
As a result, the throughput achieved by the reliability protocols may suffer as some packets are unnecessarily transmitted due to lost ACKs. However, the receiver buffer requirement is not affected by lost ACKs as the packet corresponding to the lost ACK is buffered at the receiver like every other packet.

Assume that the ACK packets are lost with probability $p_a$. At the sender, a packet is considered lost if either the packet or its ACK is lost.
With \srarq, the sender retransmits the same packet after every time-out until it receives an ACK for that packet. Thus, the throughput achieved by \srarq\ is $\rho = (1-p)(1-p_a) C$. With \srfec, however, after every time-out a new coded packet is transmitted that could potentially be useful for decoding the corresponding block at the receiver. Only when the corresponding block is (or will be) full rank then time-outs are wasteful. Thus, \srfec\ potentially achieves higher throughput compared to \srarq.

By using more sophisticated feedback mechanisms, however, one could improve the throughput of both protocols to $\rho = (1-p) C$, without affecting the receiver buffer requirement. One such mechanism is described next. Our current protocols incur feedback overhead of $\Theta(\log W)$ per-packet (see subsection~\ref{s:overhead}). As long as the packet length grows faster than the per-packet feedback overhead, there will not be any loss in throughput asymptotically. To this end, we show that the receiver can send $(1+\epsilon) \log W$ ACKs per each received packet, for any $\epsilon > 0$, and still maintain negligible feedback overhead if the packet size scales as $\omega(\log^2 W)$.

Suppose that the receiver sends $n$ ACKs back-to-back for every received packet. The achieved throughput is $\rho = (1-p)(1 - p_a^n) C$.
To achieve full throughput, the throughput loss $p_a^n C$ should converge to $0$ as $W$ grows to infinity. That is, $W p_a^n \rightarrow 0$, as $W \rightarrow \infty$, or equivalently, $p_a^n = o(\frac{1}{W})$, which yields,
\begin{equation}
\textstyle
    n = \log_{p_a} o(\frac{1}{W}) = \Omega(\log_{p_a} \frac{1}{W^{1+\epsilon}}),
\end{equation}
for any $\epsilon > 0$.

Specifically, transmitting $(1+\epsilon)\log W$ redundant ACKs ensures asymptotically full throughput for $p_a < 2^{-\frac{1}{1+\epsilon}}$. For instance, for $p_a < 0.5$, which is the case of interest in this work (TCP does not work for loss probabilities higher than $10\%$), transmitting $\log W$ redundant ACKs suffices to achieve throughput $\rho=(1-p) C$, as $W \rightarrow \infty$.

\subsection{Dependent Coded Packets}
In the preceding sections, we assumed that the sender keeps track of
the coded packets transmitted to the receiver in order to avoid
transmitting dependent coded packets. Alternatively, the sender might
simply generate and transmit coded packets regardless of their
dependency. This approach, while having the same buffer requirement,
comes with a throughput penalty, as there is a non-zero
probability that the received packets at the receiver are not
independent. In this subsection, using results from the theory of
random matrices~\cite{blake06random}, we characterize the throughput
loss incurred due to dependent packets.

Let $N_i^*$ denote the number of transmitted coded packets until the
receiver receives $B$ \emph{independent} coded packets for coding
block $i$. Let $N^*_{ik}$ denote the number of transmissions until
the receiver receives the $k$-th independent coded packet given that
it has already received $k-1$ independent coded packets for the
coding block. We have,
\begin{equation}
    N_i^* = N^*_{i1} + N^*_{i2} + \cdots + N^*_{iB}
    \eqend
\end{equation}
There are $B$ packets in a coding block, and hence, there are $2^B$ subsets of packets that can be XORed to create coded packets. While waiting to receive the $k$-th independent coded packet, the receiver has already received $k-1$ independent coded packets. Therefore, $2^{k-1}$ packets out of the $2^B$ possible coded packets will not be independent of the already received $k-1$ packets. Let $\alpha_k$ denote the probability that a received packet
is independent of the existing $k-1$ packets. The following relation
holds,
\begin{equation}
    \alpha_k = \frac{2^B - 2^{k-1}}{2^B} = 1 - 2^{k-B-1}
    \eqend
\end{equation}
A packet (dependent or independent) is received with probability
$1-p$ at the receiver. Therefore, $N^*_{ik}$ is geometrically
distributed with success probability $q_k = \alpha_k (1-p)$. It then
follows that,
\begin{equation}
\label{e:E-N-RC}
\textstyle
    \E{N^*_i} = \sum_{k=1}^{B}  \frac{1}{q_k}
          = \frac{1}{1-p} \Big( B + \sum_{k=1}^{B} \frac{1}{2^k - 1} \Big),
\end{equation}
where the summation term quickly converges to a constant ($\approx
1.606695$) even for small values of $B$. This means that, on average,
the receiver needs only two extra packets in order to be able to
decode a block of packets. Specifically, let $\delta$ denote the
number of \emph{extra} packets (dependent or independent) that will be
delivered to the receiver. Using our earlier argument, the probability that a newly arrived
packet is independent of the previously received $k-1$ independent
packets is given by,
\begin{equation}
    \beta_k = \frac{2^B - 2^{k-1}}{2^{B+\delta}} = 1 - 2^{k-B-\delta-1}
    \eqend
\end{equation}
Thus, the probability of receiving $B$ independent packets out of the
$B+\delta$ received packets, denoted by $\pi(\delta)$, is given by,
\begin{equation}
\textstyle
    \pi(\delta) = \prod_{k=1}^{B} \beta_k = \prod_{k=1}^{B} (1 -
    2^{-k-\delta}),
\end{equation}
where, we have $ \pi(\delta) \ge 1 - \sum_{k=1}^{B}
        2^{-k-\delta} \ge 1 - 2^{-\delta}$.
By taking $\delta = \Theta(\log B)$, we have $\pi(\delta) \rightarrow
1$, as $B \rightarrow \infty$. Next, we compute the throughput loss
because of dependent packets. Let $\Delta \rho$ denote the throughput
loss, we have:
\begin{equation}
\textstyle
    \Delta \rho = \frac{M \cdot \delta}{(1-p) C} = \frac{R}{1-p} \frac{M \cdot \delta}{W}
    \eqend
\end{equation}
Therefore, the throughput loss is given by,
\begin{equation}
    \Delta \rho =
        \begin{cases}
            \frac{R}{1-p} \frac{\log B}{B} = \Theta(1), & \text{if $B = \Theta(1)$},\\
            \frac{R}{1-p} \frac{\log W}{W} = o(1), & \text{if $B = \Theta(W)$} \eqend\\
        \end{cases}
\end{equation}

\subsection{Protocol Overhead}
\label{s:overhead}
Our results for the
re-sequencing buffer occupancy can be used as guidelines for
dimensioning the sender and receiver buffer space. Specifically, for
\srarq\ and \srfec\ with $B = \Theta(1)$, the proper amount of
buffer space at the sender and receiver is $\Theta(W \log W)$. On
the other hand, when $B = \Theta(W)$ under \srfec, the buffer space
at the sender and receiver should be scaled as $\Theta(W)$.

\subsubsection{Data Channel Overhead}
Both protocols require data packets in the sender and receiver
buffers to have unique sequence numbers. Thus, they both incur the
\emph{same} overhead in the forward direction from the sender to the
receiver. Moreover, the packet header overhead due to sequence
numbers is given by $\Theta(\log W)$.

\subsubsection{Feedback Channel Overhead}
Every packet has to be individually ACKed in \srarq. Thus, each ACK
packet carries the sequence number of the packet it is acknowledging.
In \srfec, however, each ACK carries the sequence number of the
corresponding coding block only. Thus, while both protocols send one
ACK per every packet received at the receiver, the header overhead
can be different for each protocol. Specifically, for \srarq\ and
\srfec\ with $B = \Theta(1)$, the per ACK overhead is in
$\Theta(\log W)$. On the other hand, with $B = \Theta(W)$, \srfec\
incurs only $\Theta(1)$ header overhead per ACK.

\subsection{Scaling Packet Length}
To amortize the header overhead of $\Theta(\log W)$, the packet length should also scale with the window size as $\omega(\log W)$. This has implications for both the packet loss probability and throughput, as discussed below.

\subsubsection{Effect on Packet Loss Probability}
Let $L$ denote the packet length and $p_e$ denote the bit error rate (BER) at the physical layer. For the ease of exposition, we ignore the intermediate protocol layers and assume that application packets are passed directly to the physical layer for transmission. The packet loss probability $p$ is then given by,
\begin{equation}
    p = 1 - (1 - p_e)^L = 1 - e^{L \ln (1- p_e)} \approx 1 - e^{- L p_e}
    \eqend
\end{equation}
Thus, as the packet length increases, the packet error probability approaches $1$ exponentially fast. To compensate for the increased packet length in order to keep the packet error probability constant, one can decrease the coding rate at the physical layer by adding more error control data to each packet. The side effect of the decreased coding rate at the physical layer is a lower per-flow capacity at the application layer. Nevertheless, the asymptotic performance of \srarq\ and \srfec\ in terms of buffer requirement remains \emph{unchanged}.

\subsubsection{Effect on Throughput}
Consider a block coding scheme at the physical layer. Define the coding rate $r$ as $r=k/n$, where $n$ and $k$ ($n \ge k$) denote the length of coding blocks and information messages respectively. Since packets are directly passed to the physical layer, each message corresponds to a packet, \ie\ $k = L$. For general classes of codes, it has been shown that~\cite{poor10finite} the maximal coding rate achievable at block length $n$ with error probability $p$ is closely approximated by,
\begin{equation}
    \textstyle
    r = \frac{L}{n} \approx C - \sqrt{\frac{V}{n}} Q^{-1}(p),
\end{equation}
where $C$ is the capacity of the underlying channel, $V$ is a characteristics of the channel referred to as channel dispersion, and $Q$ is the complementary CDF of the standard normal distribution.
Notice that, as $n \rightarrow \infty$, the maximal achievable rate approaches the channel capacity. Thus, a physical layer that is designed optimally, not only does not penalize throughput as the packet length grows but also exhibits even a better performance.

Let assume that our objective is to keep the loss probability constant at $p$ regardless of the packet length. Then, using the above approximation, it is obtained that $n=\Theta(L)$, which means that a constant throughput can be achieved for any packet length. Specifically, for any fixed packet length $L$, let assume that our objective is to achieve the throughput $\rho = (1-p) C$. It is obtained that $n = V \cdot \big( \frac{Q^{-1}(p)}{pC} \big)^2$, indicating that a fixed packet loss probability at a constant throughput can be achieved by appropriately controlling the coding block length at the physical layer.

\subsection{Implications for Protocol Design}
The scaling results for various performance metrics are summarized in
Table~\ref{t:summary}. The conclusion is that \srfec\ with $B =
\Theta(W)$ asymptotically outperforms \srarq\ in terms of buffer
requirement, delay and protocol overhead.

\section{Related Work}
\label{s:related}

The analysis of selective repeat protocols has received significant
attention in the literature. Our primary goal in this paper is to
characterize the \emph{asymptotic} performance of selective repeat
with ARQ and FEC, which has not been considered in the literature. A
summary of some related works follows.

\noindent\textbf{Selective Repeat ARQ:} Some of the early work on the
analysis of the selective repeat protocol with ARQ can be found
in~\cite{towsley79,konheim80arq,anagnos86,ros89}. These classic works
have been extended in several recent papers to consider more general
systems (\eg\ Markovian error models and packet arrival
rates)~\cite{lu93,fantacci96arq,krunz00markov}, specifically in
cellular networks in which selective repeat is typically implemented
between base stations and user devices. All these works have
considered exact or approximate analysis of selective repeat with
ARQ.

\noindent\textbf{Selective Repeat Hybrid ARQ:} Hybrid ARQ schemes
combine ARQ and FEC mechanisms similar to \srfec\ considered in this
paper. The difference is that in hybrid ARQ, the coding rate is
pre-determined based on the transmission environment. For instance, a
block coding scheme such as Reed-Solomon codes is used to generate a
fixed number of coded packets in order to cope with packet losses. In
\srfec, however, coded packets are generated on-demand until the
receiver successfully decodes the original packets (\ie\ rateless
coding). Several examples of analytical work on selective repeat with
hybrid ARQ can be found in~\cite{kallel90arq,zorzi08arq} and
references therein. Note that none of these works has considered the
asymptotic behavior of hybrid ARQ schemes.

\noindent\textbf{Packet Reordering Delay:} The analysis
of~\cite{tse05sequencing} considers reliable data transmission but
assumes that variable network delay is the cause of packet
re-ordering while the network is perfectly reliable in the sense that
no packets are lost. Along the same line,
in~\cite{zheng2010multipath}, it is assumed that the cause of
variable delay is multi-path routing in the network. Their focus is
on capturing the effect of packet reordering on receiver buffer,
which is quite different from the lossy scenario considered in this
paper.

\noindent\textbf{Asymptotic Analysis:}
A large-deviation analysis of receiver buffer behavior under
selective repeat was presented in~\cite{turck2010asymp}. Additionally, using
standard results on coupon collector problem, the authors provide an
asymptotic analysis of the receiver buffer behavior for ARQ
mechanism, which is valid only when the loss rate is \emph{very high}, as
opposed to our results that characterize receiver buffer behavior
under arbitrary loss rates for both ARQ and FEC mechanisms.

\section{Conclusion}
\label{s:conc}

In this paper, we studied reliable data transfer in high-speed
networks, where a large number of flows can be accommodated at high
transmission speeds. We focused on two sliding window mechanisms
called \srarq\ and \srfec\ based on ARQ and FEC, and characterized
their performance in the asymptotic regime of large window sizes.
Specifically, we showed that, while asymptotically achieving equal
throughput, \srfec\ with a large coding block size achieves $\log W$
improvement in buffer requirement and delay compared to \srarq.
However, \srfec\ with a small coding block size is no advantageous
compared to \srarq. It would be interesting to study the performance
of these protocols when the available buffer space is limited.

%\bibliographystyle{IEEEtran}
%\bibliography{IEEEabrv,ref,scalable}

\end{document}